\documentclass[11pt]{article}

\usepackage[margin=1.1in]{geometry}
\usepackage{amsmath,amssymb,amsthm}
\usepackage{algorithm}
\usepackage{algpseudocode}
\usepackage{enumitem}
\usepackage[colorlinks=true,linkcolor=blue,citecolor=blue,urlcolor=blue]{hyperref}

\newtheorem{theorem}{Theorem}[section]
\newtheorem{lemma}[theorem]{Lemma}

\newtheorem{corollary}[theorem]{Corollary}

\theoremstyle{definition}
\newtheorem{definition}[theorem]{Definition}

\newcommand{\Z}{\mathbb{Z}}
\newcommand{\N}{\mathbb{N}}

\newcommand{\Mtil}{\widetilde{M}}

\newcommand{\D}{\mathcal{D}}

\usepackage{mathtools}
\usepackage{cleveref}
\usepackage[normalem]{ulem}

\title{Fake It Until You Make It:\\Linear-Time Verification of Rings and Fields}

\date{}
\makeatletter
\let\old@fnsymbol\@fnsymbol
\renewcommand{\@fnsymbol}[1]{%
  \ifnum#1=1\relax
    \ensuremath{\dagger}%
  \else
    \old@fnsymbol{#1}%
  \fi
}
\makeatother
\author{Youlong Ding\thanks{The Hebrew University of Jerusalem, Israel. Email: \texttt{youlong.ding@mail.huji.ac.il}.}}

\begin{document}
\maketitle

\begin{abstract}
We consider the following problems: Given two $n \times n$ tables defining binary operations $+$ and $\cdot$ on a set $S$ of $n$ elements,  decide whether $(S,+,\cdot)$ forms a ring or, respectively, a field.
Recently, Dudek, Fischer, Gokaj, Jin, K\"unnemann, Mao, and Redzic (STOC 2026) obtained the following two (near-)optimal results:
\begin{itemize}
    \item A randomized $O(n^2\log(1/\delta))$-time algorithm for verifying rings.
    \item A deterministic $O(n^2)$-time algorithm for verifying fields.
\end{itemize}
Their algorithms build on machinery of Evra, Gadot, Klein, and Komargodski (FOCS 2024), which relies on Classification of Finite Simple Groups (CFSG).

In this work, we give a deterministic $O(n^2)$-time algorithm for ring verification, resolving the deterministic complexity of this problem. As a corollary, we also obtain a deterministic $O(n^2)$-time algorithm for field verification. 
Our algorithms are elementary and avoid CFSG machinery entirely.
\end{abstract}

\newpage


\section{Introduction}\label{sec:intro}

We consider fundamental verification questions for algebraic structures consisting of a ground set $S$ together with one or more operations $+$, $\cdot$, etc. on $S$. For example: does $(S, +)$ form a group? Does $(S,+,\cdot)$ form a ring? Does $(S,+,\cdot)$ form a field? In this paper, we consider the finite setting, where $|S|=n$ and each binary operation is represented explicitly by its $n\times n$ Cayley table.

Since even a single binary operation is represented by an $n\times n$ table, the input size is already $\Theta(n^2)$. It is obvious that every algorithm for the above problem must run in time $\Omega(n^2)$, i.e., linear in the input size.
The question is then about an optimal upper bound: 
\begin{center}
    Can we design linear-time algorithms for verifying algebraic structures?
\end{center}

Since verifying the existence of identity/inverse can both be trivially done in time $O(n^2)$, answering such questions essentially boils down to deciding whether certain algebraic identities hold for the given tables, such as the associative law
$$
(a\cdot b) \cdot c=a\cdot (b\cdot c)
$$
or the distributive law
$$
a\cdot (b+c)=a\cdot b+a \cdot c.
$$

Interestingly, even verifying the associativity of a given operation turns out to be non-trivial. A naive verification checks all triples, resulting in an $O(n^3)$ running time. An observation attributed to F.~W.~Light in 1949 (see \cite{CliffordPreston1961}) yields a deterministic $O(n^2\log n)$ test when a logarithmic-size generating set is available. Rajagopalan and Schulman~\cite{Siamcomp:RajagopalanS00} gave a surprising randomized algorithm that checks associativity in time
$O(n^2\log(1/\delta))$ with one-sided error $\delta$.

A recent breakthrough result of Evra, Gadot, Klein, and Komargodski~\cite{FOCS:EvraGKK24} gives a \emph{deterministic} $O(n^2)$-time algorithm for verifying groups. Their idea is to find a so-called \emph{basis} $B$ for $S$ of size $O(\sqrt{n})$ in $O(n^2)$ time, where a basis is a set $B$ such that $B+B=S$, i.e., every element of $S$ can be written as the sum of two elements of $B$. Then they show that associativity testing reduces to checking that all five ways of parenthesizing
$w+x+y+z$ agree for every $w,x,y,z\in B$, which takes $O(|B|^4)=O(n^2)$ time.

More recently, Dudek, Fischer, Gokaj, Jin, K\"{u}nnemann, Mao, and Redzic~\cite{STOC:DudekFGJKMR26} extended this basis approach to rings and fields by verifying distributive laws through partially restricted checks. Namely, to verify
$
a\cdot(b+c)=a\cdot b+a\cdot c,
$
they show that it suffices to restrict two of the variables to the basis $B$, resulting in $O(|B|^2n)=O(n^2)$ time. Together with the basis approach for group verification, this yields an $O(n^2)$-time algorithm for verifying fields.

The subtlety arises for rings. Since $(S,\cdot)$ is not necessarily a group, the basis-finding approach is inapplicable to the multiplication operation. To verify associativity of multiplication, \cite{STOC:DudekFGJKMR26} therefore resorts to the randomized Rajagopalan--Schulman test~\cite{Siamcomp:RajagopalanS00}, resulting in only a randomized $O(n^2\log(1/\delta))$-time algorithm for verifying rings.

A second drawback of \cite{STOC:DudekFGJKMR26} is that their algorithms for rings and fields inherit the complexity of the basis-finding algorithm in \cite{FOCS:EvraGKK24}, which is very complicated and relies on the heavy machinery of the Classification of Finite Simple Groups (CFSG)~\cite{aschbacher2004status}.

In this paper, we improve upon \cite{STOC:DudekFGJKMR26} by giving an optimal deterministic algorithm for verifying rings. Moreover, our algorithm is elementary and does not rely on CFSG.

\begin{theorem}[Main theorem; informal]\label{thm:intro-main}
Given a set $S$ of $n$ elements, together with binary operations $+$ and $\cdot$ on $S$, there exists a deterministic algorithm that decides whether $(S,+,\cdot)$ forms a ring in time $O(n^2)$.
\end{theorem}

As a direct corollary, we also obtain an $O(n^2)$-time algorithm for field verification without the use of CFSG,\footnote{Verifying fields reduces to verifying rings by additionally checking the existence of a multiplicative identity, multiplicative inverses for all nonzero elements, and commutativity of multiplication, all of which can be checked straightforwardly in $O(n^2)$ time.} yielding a substantially simpler alternative to the algorithm of \cite{STOC:DudekFGJKMR26}.

\begin{corollary}[Informal]\label{cor:intro-field}
Given a set $S$ of $n$ elements, together with binary operations $+$ and $\cdot$ on $S$, there exists a deterministic algorithm that decides whether $(S,+,\cdot)$ forms a field in time $O(n^2)$.
\end{corollary}


\paragraph{Related Work.}
We briefly mention related work beyond that discussed in \cite{FOCS:EvraGKK24} and \cite{STOC:DudekFGJKMR26}.
A recent work \cite{Beals26} shows how to eliminate the use of CFSG from linear-time associativity verification for groups.
We note that this does not remove the dependence on CFSG from the ring-verification algorithm of \cite{STOC:DudekFGJKMR26}, as their test for distributivity still relies on CFSG-dependent machinery.



\section{Technical Overview}

\begin{center}
\textbf{TL;DR:} \emph{A biadditive multiplication is determined by its
values on additive generators, and its associativity can be checked on
triples of those generators.}
\end{center}

At a conceptual level, our algorithm rests on the above two classical facts about
biadditive multiplications. The main challenge is to turn these algebraic facts
into an \underline{efficient verification} procedure. We first explain the two ideas that
make this possible at a high level, and then give the technical details.

\paragraph{Idea 1 (for deterministic optimality): A Dependency Reversal.}
A natural verification order is to first verify the associativity of $(S, \cdot)$ and then verify distributivity.
This is the order used in \cite{STOC:DudekFGJKMR26} and is the key reason why they need to resort to a randomized associativity test.
Our method instead uses the opposite order, based on the following key observation: once distributivity is established, verifying $(ab)c = a(bc)$ (i.e., the associativity of $(S, \cdot)$) reduces to verifying it with $a,b,c$ all restricted to generators of $(S, +)$.
Since every group has a generating set of size at most $\log n$, associativity requires only
$O(\log^3 n)$ checks rather than the naive $O(n^3)$ checks.
Though simple in hindsight, this key observation already allows us to upgrade the $O(n^2\log(1/\delta))$-time randomized ring verifier in \cite{STOC:DudekFGJKMR26} to an optimal deterministic $O(n^2)$-time verifier.

\paragraph{Idea 2 (for removing the machinery of CFSG): Faking It Until You Make It.}
We then seek to remove the heavy machinery of CFSG used in \cite{STOC:DudekFGJKMR26} for verifying the distributivity of $(S,+,\cdot)$.\footnote{Indeed, their algorithm also uses CFSG to verify the associativity of $(S,+)$, but this can now be eliminated by \cite{Beals26}.}
The idea is to use a candidate cyclic decomposition of $(S,+)$ to generate (or ``fake'') a multiplication table on $S\times S$ \underline{in $O(n^2)$ time}\footnote{That is, each entry of the table should be generated in amortized $O(1)$ time.} that is expected to be distributive by construction. We then compare this faked table with the input multiplication table. The key property is that the input multiplication satisfies the distributive laws if and only if it agrees with the faked table.

In what follows, we provide more details. We begin by explaining how to verify that $(S,+)$ is an abelian group. This part serves two purposes. First, it provides a simpler alternative to \cite{Beals26}, specialized to abelian groups, for eliminating the use of CFSG. Second, the machinery developed here will later be reused to verify distributivity.

\paragraph{Verifying an Abelian Group by Checking an Isomorphism.}
A key structural property of abelian groups that we rely on is that every finite abelian group $(S,+)$ is isomorphic to a direct product of cyclic groups:
\[
  (S,+)\;\cong\;\Z_{d_1} \times \cdots\times \Z_{d_k},\qquad
  \textstyle\prod_i d_i=n,\quad k\le\log_2 n ,
\]
with independent generators $g_1,\dots,g_k$. In other words, there exists an isomorphism
\[
\sigma\colon S\to \Z_{d_1}\times\cdots\times\Z_{d_k},
\qquad
a\mapsto (a_1,\dots,a_k).
\]

We construct a mapping $\sigma$ by running a standard cyclic-decomposition algorithm for abelian groups \cite{Vikas1996,Kavitha2007} in time $\widetilde{O}(n)$.
The algorithm will output $g_1, \dots, g_k$, $d_1, \dots, d_k$, from which we can construct a mapping $\sigma$.
This requires some explanation, since \uline{the input is not guaranteed to be an abelian group}: what does it mean to run an algorithm intended for abelian groups on an arbitrary input? Following \cite{FOCS:EvraGKK24}, we emulate the decomposition algorithm under a prescribed time budget. If the input is indeed an abelian group, the algorithm is guaranteed to return an output within this budget. If the budget is exceeded, we simply reject.
We refer the reader to \Cref{sec:abelian} for more details. We then continue with the constructed map $\sigma$. The following two statements are then equivalent:
\begin{itemize}
\item $(S,+)$ is an abelian group.
\item The constructed mapping $\sigma$ is an isomorphism; namely, $\sigma$ is both a bijection and a group homomorphism.
\end{itemize}
Therefore, verifying that $(S,+)$ is an abelian group (which also verifies its associativity as a byproduct) reduces to two tasks. First, we check that $\sigma$ is a bijection, which can be done in $\widetilde{O}(n)$ time. Second, we check that $\sigma$ is a group homomorphism: for every $a,b\in S$, we verify
\[
\sigma(a+b)=\sigma(a)+\sigma(b).
\]
However, recall that $\sigma(a)$ and $\sigma(b)$ are $k$-dimensional vectors, where $k=O(\log n)$, so computing their sum takes $O(\log n)$ time. Checking the homomorphism condition for all pairs $a,b$ therefore takes $O(n^2\log n)$ time, which exceeds our target of $O(n^2)$.

\paragraph{Rescue: Checking the Homomorphism of $\sigma^{-1}$ via Incremental Computation.}
Since $\sigma$ is a (verified) bijection, checking that $\sigma$ is a group homomorphism is equivalent to checking that $\sigma^{-1}$ is a group homomorphism. Denote
$$
D:=\Z_{d_1}\times\cdots\times\Z_{d_k}.
$$
Then for every
$$
\mathbf{a}=(a_1,\dots,a_k),\qquad
\mathbf{b}=(b_1,\dots,b_k)\in D,
$$
we need to verify
$$
\sigma^{-1}(\mathbf{a}+\mathbf{b}) = \sigma^{-1}(\mathbf{a})+\sigma^{-1}(\mathbf{b}).
$$
Now the right-hand side can be computed in $O(1)$ time using one lookup in the input addition table, but the bottleneck becomes the left-hand side:  computing the sum $\mathbf{a}+\mathbf{b}$ takes $O(k)=O(\log n)$ time.

We avoid this overhead through incremental computation. For each fixed $\mathbf{a}\in D$, we enumerate all $\mathbf{b}\in D$ in counting order. During this enumeration, we maintain a running value
$$
\mathbf{z}\gets \mathbf{a}+\mathbf{b},
\qquad\text{i.e.,}\qquad
(z_1,\dots,z_k)\gets (a_1,\dots,a_k)+(b_1,\dots,b_k).
$$
Whenever the enumeration performs an increment of one coordinate of $\mathbf{b}$, say $b_i$ (either as the usual increment or as part of a carry), we update the corresponding coordinate $z_i$ in exactly the same way. Thus, instead of recomputing $\mathbf{a}+\mathbf{b}$ from scratch for each $\mathbf{b}$, we maintain it incrementally throughout the enumeration.
Therefore the enumeration of all $n$ tuples of $\mathbf{b}$ uses only $O(n)$ such increments in total. Thus, for each fixed $\mathbf{a}$, all values of $\mathbf{a}+\mathbf{b}$ can be maintained in $O(n)$ total time.
Consequently, the homomorphism condition for all $n^2$ pairs can be verified in $O(n^2)$ total time.

So far, we have shown how to verify that $(S,+)$ is an abelian group. It remains to verify the additional structure required for a ring: distributivity and the associativity of $(S,\cdot)$. As described in Idea 1, once distributivity has been established, associativity needs to be checked only on triples of generators of $(S,+)$, and hence requires only $O(\log^3 n)$ checks.

Thus, the remaining task is to verify distributivity, say
$$
a\cdot(b+c)=a\cdot b+a\cdot c.
$$
A naive verification enumerates all $O(n^3)$ triples $(a,b,c)$. \cite{STOC:DudekFGJKMR26} extended basis-finding approach in \cite{FOCS:EvraGKK24} to verify distributivity through partially restricted checks. However, finding such basis requires heavy machinery of CFSG.
In the remainder of this section, we present an elementary algorithm for verifying distributivity that avoids CFSG entirely.

\paragraph{Verifying Distributivity by Faking a Multiplication Table $\Mtil$.}
Our idea is that instead of checking whether the given multiplication table $M$ satisfies the $O(n^3)$ equations imposed by distributivity law, we assume distributivity and ``fake'' a multiplication table $\Mtil$ from scratch, and then check whether $\Mtil$ coincides with $M$. 

Let $g_1,\dots,g_k$ be generators of $(S,+)$.\footnote{This is well-defined since we have already verified that $(S,+)$ is an abelian group.} We first define
$$
\Mtil(g_i,g_j):=M(g_i,g_j)
$$
for every pair of generators. The key fact we rely on is that, once these values are fixed, distributivity determines all other entries of $\Mtil$. Indeed, writing $a$ and $b$ as linear combinations of the generators, $\Mtil(a,b)$ is determined as the corresponding linear combination of the values $\Mtil(g_i,g_j)$, with coefficients given by the coordinates of $a$ and $b$.

Some explanation is in order. \uline{The values $M(g_i,g_j)$ need not be viewed as any ``ground truth'' for $g_i\cdot g_j$, nor do we need them to be.} We simply use these entries of $M$ as the starting data from which a distributive multiplication table $\Mtil$ is reconstructed. If the resulting table $\Mtil$ agrees with $M$ everywhere, then $M$ itself must satisfy the distributive laws. Conversely, if $M$ already satisfies the distributive laws, then distributivity forces every entry of $M$ to be exactly the value reconstructed in $\Mtil$, and hence $M=\Mtil$. Thus, verifying distributivity reduces to checking whether the faked table $\Mtil$ coincides with the input table $M$.\footnote{The same philosophy already appeared in our verification of the abelian group structure. When constructing the candidate mapping $\sigma$, we do not need to regard the entries $A(a,b)$ as ground truth for $a+b$. What matters is that, if the resulting $\sigma$ is verified to be an isomorphism with respect to $A$, then $A$ itself necessarily defines an abelian group.}

More formally, we fake the table $\Mtil$ as follows:
\begin{enumerate}
\item Define $g_i\coloneqq \sigma^{-1}(e_i)$ for $i\in[k]$, where $e_i$ is the $i$-th basis vector.
\item For $a,b\in S$, writing
$$
\sigma(a)=(a_1,\dots,a_k),\qquad
\sigma(b)=(b_1,\dots,b_k),
$$
define
$$
\Mtil(a,b)\coloneqq
\sum_{i,j} a_i b_j(g_i\cdot g_j),
$$
where $a_i\in \{0,\dots,d_i-1\}$ and $b_j\in \{0,\dots,d_j-1\}$ are the coordinates of $a$ and $b$, respectively.
\end{enumerate}

\paragraph{Incremental Computation, Second Encounter.}
Again, computing each entry $\Mtil(a,b)$ directly from the above formula is too expensive: it involves forming the coefficients $a_i b_j$, taking the corresponding integer multiples of $g_i\cdot g_j$, and summing over all pairs $(i,j)$. Doing this independently for all $n^2$ entries would therefore exceed our $O(n^2)$ time budget.

We again overcome this bottleneck through incremental computation. The computation proceeds in two stages. First, for each generator $g_i$, we generate the entire row
$$
\phi_i(b)\coloneqq \Mtil(g_i,b).
$$
We maintain a running $b\in S$ in the counting order as before (implemented through the map $\sigma^{-1}$). Whenever the enumeration increments the $j$-th coordinate of $b$, we update
$$
\phi_i(b)\gets \phi_i(b)+(g_i\cdot g_j).
$$
Thus, each elementary increment requires only one addition in $(S,+)$. Since there are only $O(n)$ elementary increments in one complete enumeration, all $k=O(\log n)$ generator rows can be generated in $O(n\log n)$ total time.

Next, we generate all remaining rows of $\Mtil$. We maintain a running $a\in S$ in the counting order while maintaining the current row
$$
r(\cdot)\coloneqq \Mtil(a,\cdot).
$$
Whenever the enumeration increments the $i$-th coordinate of $a$, we update for all $b \in S$
$$
\Mtil(a+g_i,b) \gets \Mtil(a,b)+\Mtil(g_i,b)
=r(b)+\phi_i(b).
$$
Hence, the new row is obtained simply by adding the precomputed row $\phi_i$ entrywise to $r$. This takes $O(n)$ time per elementary increment of $a$. Since the entire enumeration of $a$ involves only $O(n)$ such increments, all $n$ rows of $\Mtil$ are generated in $O(n^2)$ total time.
In this way, although an individual entry of $\Mtil$ is not computed from its defining formula in constant time, the entire table is generated in $O(n^2)$ time, or amortized $O(1)$ time per entry.




\section{Preliminaries}\label{sec:prelim}

\paragraph{Model and conventions.}
We work in the word-RAM model with words of $O(\log n)$ bits. The input consists of
$n$ and two read-only tables $A,M:[n]\times[n]\to[n]$; reading one cell is a
\emph{probe}. We write $x+y:=A(x,y)$ and $xy:=x\cdot y:=M(x,y)$. The algorithm
rejects as soon as it reads a value outside $[n]$. For $m\in\N$ and an element $x$ of
a (certified) abelian group we write $m\cdot x$ for the $m$-fold sum, computable with
$O(\log m)$ additions by binary doubling. All logarithms are base $2$.

\begin{definition}[Abelian group]
$([n],A)$ is an \emph{abelian group} if $+$ is associative and commutative, there exists an element $0\in[n]$ such that $x+0=0+x=x$ for every $x\in[n]$, and for every $x\in[n]$ there exists an element $-x\in[n]$ such that
$
x+(-x)=(-x)+x=0.
$
\end{definition}

\begin{definition}[Ring]
$([n],A,M)$ is a \emph{ring} if $([n],+)$ is an abelian group, $\cdot$ is
associative, and $x(y+z)=xy+xz$ and $(y+z)x=yx+zx$ for all $x,y,z$.
\end{definition}

\begin{definition}[Field]
$([n],A,M)$ is a \emph{field} if it is a ring, $\cdot$ is commutative, and there exists an element $1\neq 0$ such that $1x=x1=x$ for every $x\in[n]$, where $0$ denotes the additive identity. Moreover, for every $x\neq 0$, there exists $x^{-1}\in[n]$ such that
$
xx^{-1}=x^{-1}x=1.
$
\end{definition}

\begin{lemma}[Cyclic decomposition of finite abelian groups]
\label{lem:abelian-decomposition}
Let $(S,+)$ be an abelian group of order $n$, given by its Cayley table $(n, A)$. There is a deterministic algorithm that, in time $\widetilde{O}(n)$, computes elements $g_1,\ldots,g_k\in S$ and integers $d_1,\ldots,d_k\ge 2$ such that
$$
(S,+)\cong \mathbb{Z}_{d_1}\times\cdots\times\mathbb{Z}_{d_k},
\qquad
\prod_{i=1}^k d_i=n,
$$
where $g_i$ generates the $i$-th cyclic factor. In particular, every $x\in S$ admits a unique representation
$
x=\sum_{i=1}^k x_i g_i,
\qquad
x_i\in{0,\ldots,d_i-1}.
$
Moreover, $k\le \log_2 n$.
\end{lemma}

\begin{proof}
This follows from the standard algorithms for computing cyclic decompositions of finite abelian groups; see \cite{Vikas1996,Kavitha2007}.
\end{proof}

\paragraph{Tuple groups.}
For $d_1,\dots,d_k\ge 2$ with $\prod_i d_i=n$ let
$\D=\D(d_1,\dots,d_k):=\Z_{d_1}\times\cdots\times\Z_{d_k}$ with componentwise
addition $\oplus$; its identity is $\mathbf 0$ and its unit vectors are
$e_1,\dots,e_k$. We identify tuples with their mixed-radix ranks in $[0,n)$ via
$t\mapsto t_1+t_2 d_1+t_3 d_1d_2+\cdots$; a tuple is stored as its rank plus its
digit array, so equality tests are $O(1)$. Note $k\le\log n$ since every $d_i\ge2$.
When $n=1$ we set $k=0$ and $\D$ is the trivial group.

The engine behind all our streaming passes is the elementary observation that
counting visits every tuple while changing few digits.

\begin{lemma}[Odometer]\label{lem:odometer}
Enumerate $\D$ in mixed-radix counting order $t^{(0)}=\mathbf 0,t^{(1)},\dots,
t^{(n-1)}$. Each transition $t^{(m)}\to t^{(m+1)}$ decomposes into a sequence of
\emph{elementary increments}, each of the form ``add $e_i$'' (that is, digit
$i\mathrel{+}=1$ modulo $d_i$): first $e_1$ is added; whenever digits $1,\dots,i$
have all just wrapped to $0$, $e_{i+1}$ is added as well. The enumeration visits
every tuple exactly once, the total number of elementary increments is at most $2n$,
and the enumeration is driven in $O(1)$ amortized time per elementary increment.
\end{lemma}

\begin{proof}
Counting from $0$ to $n-1$ in mixed radix visits each rank once, and the standard
carry rule is exactly the stated sequence of digit increments: adding $1$ to digit
$i$ maps $d_i-1$ to $0$, which is precisely addition of $e_i$ in $\Z_{d_i}$, so no
information is lost by the wrap and the next carry propagates as an independent
addition of $e_{i+1}$. Digit $i$ is incremented once per $\prod_{j<i}d_j$ ranks,
hence at most $n/2^{\,i-1}$ times in total; summing the geometric series bounds the
number of elementary increments by $2n$. Maintaining the digit array and detecting
wraps is $O(1)$ per increment.
\end{proof}

\section{Verifying Abelian Groups in Linear Time}\label{sec:abelian}

In this section, we give a deterministic $O(n^2)$-time algorithm (Algorithm~\ref{alg:verifyabelian}) for verifying whether the input table $A$ defines an abelian group.
Our algorithm is considerably simpler than the approaches implied by \cite{FOCS:EvraGKK24,Beals26}.

We use the standard cyclic-decomposition algorithm for abelian groups from Lemma~\ref{lem:abelian-decomposition}. Since the input is not promised to be an abelian group, we use this algorithm only to construct a \emph{candidate} coordinate system. More precisely, following \cite{FOCS:EvraGKK24}, we run the decomposition algorithm under a prescribed time budget that is guaranteed to suffice whenever the input is indeed an abelian group. If the algorithm fails to terminate within this budget, we simply reject. Any output produced within the budget is treated as untrusted and is subsequently certified by a direct isomorphism check.

\begin{algorithm}[t]
\caption{\textsc{VerifyAbelian}$(A)$}
\label{alg:verifyabelian}
\begin{algorithmic}[1]
\State Run the cyclic-decomposition algorithm of Lemma~\ref{lem:abelian-decomposition} under a $\widetilde O(n)$ time budget
\State \textbf{reject} on timeout or malformed output; otherwise obtain $(d_1,\dots,d_k)$ and candidate generators $\widehat g_1,\dots,\widehat g_k$
\State \textbf{reject} unless $d_i\ge 2$ for all $i$ and $\prod_i d_i=n$
\State Find the unique $\widehat 0$ satisfying $A(\widehat 0,\widehat 0)=\widehat 0$; \textbf{reject} if it does not exist uniquely
\State $x\gets\widehat 0$ and $\tau(\mathbf 0)\gets\widehat 0$
\For{each subsequent tuple $t\in\D$ in odometer order}
\State for every elementary increment $e_i$ in the transition to $t$, update $x\gets A(x,\widehat g_i)$
\State $\tau(t)\gets x$
\EndFor
\State \textbf{reject} unless $\tau:\D\to[n]$ is a bijection; set $\sigma\gets\tau^{-1}$
\For{$a\in\D$}
\State $z\gets a$
\For{$b\in\D$ in odometer order}
\State \textbf{reject} unless $A\bigl(\tau(a),\tau(b)\bigr)=\tau(z)$
\State when advancing $b$, apply every elementary increment $e_j$ also to $z$
\EndFor
\EndFor
\State \textbf{accept} and output $(d_i)_i,\sigma,\tau,\ g_i:=\tau(e_i),\ 0:=\tau(\mathbf 0)$
\end{algorithmic}
\end{algorithm}

\begin{theorem}[Abelian group verification]\label{thm:abelian}
\Cref{alg:verifyabelian} decides whether an $n\times n$ table $A$ is the Cayley table of an abelian group in time $O(n^2)$. On acceptance, it outputs integers $d_1,\dots,d_k$, an isomorphism
$$
\sigma:([n],A)\longrightarrow
\D(d_1,\dots,d_k)
=\Z_{d_1}\times\cdots\times\Z_{d_k},
$$
its inverse $\tau=\sigma^{-1}$, generators
$$
g_i:=\tau(e_i),
$$
and the additive identity
$$
0:=\tau(\mathbf 0).
$$
In particular, $k\le \log n$ and every $x\in[n]$ has the unique representation
$$
x=\sum_{i=1}^k \sigma(x)_i\cdot g_i.
$$
\end{theorem}

\begin{proof}
We first prove completeness. Suppose that $A$ is the Cayley table of an abelian group. By Lemma~\ref{lem:abelian-decomposition}, the decomposition algorithm terminates within its budget and returns generators $\widehat g_1,\dots,\widehat g_k$ and integers $d_1,\dots,d_k$ realizing a cyclic decomposition
$$
([n],A)\cong \D=\Z_{d_1}\times\cdots\times\Z_{d_k}.
$$
The unique idempotent $\widehat 0$ is the additive identity. During the odometer sweep, an immediate induction shows that for every $t=(t_1,\dots,t_k)\in\D$,
$$
\tau(t)=\sum_{i=1}^k t_i\cdot\widehat g_i.
$$
Thus $\tau$ is precisely the isomorphism induced by the cyclic decomposition. In particular, it is bijective and satisfies
$$
A\bigl(\tau(a),\tau(b)\bigr)=\tau(a\oplus b)
\qquad\text{for all }a,b\in\D,
$$
so the algorithm accepts.

For soundness, suppose that the algorithm accepts. \uline{We need not trust any property of the output of the decomposition algorithm.} The final comparison directly certifies that $\tau:\D\to[n]$ is a bijection satisfying
$$
A\bigl(\tau(a),\tau(b)\bigr)=\tau(a\oplus b)
\qquad\text{for all }a,b\in\D.
$$
Hence $\tau$ is a group isomorphism from the abelian group $(\D,\oplus)$ onto $([n],A)$. Therefore $([n],A)$ is an abelian group and $\sigma=\tau^{-1}$ is the claimed isomorphism.

It remains to analyze the running time. The candidate decomposition takes $\widetilde O(n)$ time. By Lemma~\ref{lem:odometer}, constructing $\tau$ takes $O(n)$ time, and its bijectivity can also be checked in $O(n)$ time. The final comparison performs exactly one probe of $A$ for each pair $(a,b)\in\D^2$, hence $n^2$ probes in total. For each fixed $a$, the values $z=a\oplus b$ are maintained throughout the odometer sweep using only $O(n)$ elementary increments. Thus the comparison takes $O(n^2)$ time overall.

Finally, after acceptance, $g_i=\tau(e_i)$ and $0=\tau(\mathbf 0)$ are certified generators and the additive identity. Since $\tau=\sigma^{-1}$, every $x\in[n]$ satisfies
$$
x=\tau(\sigma(x))
=\sum_{i=1}^k \sigma(x)_i\cdot g_i,
$$
as claimed.
\end{proof}

\section{Verifying Rings in Linear Time}
\label{sec:ring}
\label{sec:bilinear}
\label{sec:assoc}

In this section, we give a deterministic $O(n^2)$-time algorithm for verifying
whether the input tables $(A,M)$ define a ring. The algorithm has three steps:
we first certify that $([n],A)$ is an abelian group, then verify the two
distributive laws, and finally verify associativity of multiplication on a
small set of additive generators.

Suppose first that \textsc{VerifyAbelian} has accepted $A$. By
Theorem~\ref{thm:abelian}, we then have a certified isomorphism
\[
G:=([n],+)\cong
\D:=\Z_{d_1}\times\cdots\times\Z_{d_k},
\]
with generators $g_i=\sigma^{-1}(e_i)$, additive identity $0$, and
$k\le \log n$. For $x\in G$, write
\[
\sigma(x)=(x_1,\dots,x_k),
\qquad
x=\sum_{i=1}^k x_i\cdot g_i.
\]

\begin{algorithm}[t]
\caption{\textsc{VerifyDistributive}$(M;\,\text{certified }G)$}
\label{alg:bilinear}
\begin{algorithmic}[1]
\State $c_{ij}\gets M(g_i,g_j)$ for all $i,j\in[k]$
\State \textbf{reject} unless
$d_i\cdot c_{ij}=0$ and $d_j\cdot c_{ij}=0$
for every $i,j\in[k]$
\For{$i=1,\dots,k$}
    \Comment{construct $\varphi_i(y)=\Mtil(g_i,y)$}
    \State $y\gets 0$, $v\gets 0$, and $\varphi_i[0]\gets 0$
    \For{each subsequent tuple in $\D$ in odometer order}
        \State for every elementary increment $e_j$ in the transition,
        update $y\gets A(y,g_j)$ and $v\gets A(v,c_{ij})$
        \State $\varphi_i[y]\gets v$
    \EndFor
\EndFor
\State $x\gets 0$ and $r[y]\gets 0$ for all $y\in[n]$
\State \textbf{reject} unless $M(x,y)=r[y]$ for every $y\in[n]$
\For{each subsequent tuple in $\D$ in odometer order}
    \State for every elementary increment $e_i$ in the transition,
    update $x\gets A(x,g_i)$ and
    $r[y]\gets A(r[y],\varphi_i[y])$ for every $y\in[n]$
    \State \textbf{reject} unless $M(x,y)=r[y]$ for every $y\in[n]$
\EndFor
\State \textbf{accept}
\end{algorithmic}
\end{algorithm}

\begin{theorem}[Distributivity verification]\label{thm:bilinear}
Suppose the additive table $A$ has been certified by
\textup{\textsc{VerifyAbelian}}.
Then Algorithm~\ref{alg:bilinear} accepts if and only if $M$ satisfies both
distributive laws over $([n],A)$. It runs in deterministic time $O(n^2)$.
\end{theorem}

\begin{proof}
We first define
\[
\Mtil(x,y)
:=
\sum_{i=1}^k\sum_{j=1}^k
(x_i y_j)\cdot c_{ij}.
\tag{5.1}\label{eq:rebuilt-product}
\]
Since
\[
d_i\cdot c_{ij}=0
\qquad\text{and}\qquad
d_j\cdot c_{ij}=0
\tag{5.2}\label{eq:order-conditions}
\]
for every $i,j\in[k]$, $\Mtil$ is well defined.

Suppose first that $M$ satisfies both distributive laws. Since
$d_i\cdot g_i=d_j\cdot g_j=0$,
\[
d_i\cdot c_{ij}=(d_i\cdot g_i)g_j=0,
\qquad
d_j\cdot c_{ij}=g_i(d_j\cdot g_j)=0,
\]
so the order conditions pass. Moreover, distributivity gives
\[
M(x,y)
=
M\left(\sum_i x_i\cdot g_i,\sum_j y_j\cdot g_j\right)
=
\sum_{i,j}(x_i y_j)\cdot M(g_i,g_j)
=
\Mtil(x,y).
\tag{5.3}\label{eq:M-equals-Mtilde}
\]

The algorithm computes $\Mtil$ incrementally. During the construction of
$\varphi_i$, the invariant is
\[
v=\Mtil(g_i,y):
\]
an elementary increment $e_j$ changes $y$ to $y+g_j$ and $v$ to
$v+c_{ij}=\Mtil(g_i,y+g_j)$. Hence
$\varphi_i[y]=\Mtil(g_i,y)$ for every $y$. During the outer sweep, the invariant
is
\[
r[y]=\Mtil(x,y)\qquad\text{for every }y.
\]
Indeed, an increment $e_i$ changes the row by
\[
\Mtil(x+g_i,y)=\Mtil(x,y)+\varphi_i[y].
\]
Since the odometer visits every $x$ and $M=\Mtil$, all comparisons succeed.

Conversely, suppose the algorithm accepts. The order conditions make
\eqref{eq:rebuilt-product} well defined and biadditive: changing $x_i$ by
$d_i$, or $y_j$ by $d_j$, changes the corresponding terms by zero.
The same two invariants therefore show that the algorithm constructs exactly
the table $\Mtil$. Since every row is compared with $M$ and all comparisons
succeed, $M=\Mtil$. Hence $M$ is biadditive, which is exactly the conjunction
of the two distributive laws.

For the running time, the order checks cost $O(k^2\log n)=O(\log^3 n)$.
By Lemma~\ref{lem:odometer}, constructing all $k$ generator rows costs
$O(nk)=O(n\log n)$. The outer odometer sweep has $O(n)$ elementary increments,
each requiring an $O(n)$ row update, and the row comparisons examine $n^2$
entries in total. Thus the overall running time is $O(n^2)$.
\end{proof}

\begin{algorithm}[t]
\caption{\textsc{VerifyRing}$(A,M)$}
\label{alg:main}
\begin{algorithmic}[1]
\State run \textsc{VerifyAbelian}$(A)$; \textbf{reject} if it rejects
\State obtain the certified additive generators $g_1,\dots,g_k$
\State run \textsc{VerifyDistributive}$(M;\,\text{certified }G)$;
\textbf{reject} if it rejects
\For{$i,j,\ell\in[k]$}
    \State \textbf{reject} unless
    $M(M(g_i,g_j),g_\ell)=M(g_i,M(g_j,g_\ell))$
\EndFor
\State \textbf{accept}
\end{algorithmic}
\end{algorithm}

\paragraph{Associativity reduces to additive generators.}
Once Algorithm~\ref{alg:bilinear} accepts, multiplication is biadditive. Define
the associator
\[
F(x,y,z):=(xy)z-x(yz).
\]
Biadditivity of multiplication implies that $F$ is additive in each argument.
Therefore, writing
\[
x=\sum_i x_i\cdot g_i,\qquad
y=\sum_j y_j\cdot g_j,\qquad
z=\sum_\ell z_\ell\cdot g_\ell,
\]
we have
\[
F(x,y,z)
=
\sum_{i,j,\ell}
(x_i y_j z_\ell)\cdot F(g_i,g_j,g_\ell).
\]
Consequently, multiplication is associative on all of $G$ if and only if
\[
(g_i g_j)g_\ell=g_i(g_j g_\ell)
\qquad
\text{for every }i,j,\ell\in[k].
\]
Thus associativity can be verified using only
$O(k^3)=O(\log^3 n)$ checks.

\begin{theorem}[Ring verification]\label{thm:main}
Algorithm~\ref{alg:main} accepts if and only if $([n],A,M)$ is a ring.
It runs deterministically in $O(n^2)$ time.
\end{theorem}

\begin{proof}
If $([n],A,M)$ is a ring, then \textsc{VerifyAbelian} accepts by
Theorem~\ref{thm:abelian}, \textsc{VerifyDistributive} accepts by
Theorem~\ref{thm:bilinear}, and all generator triples satisfy associativity.
Hence the algorithm accepts.

Conversely, suppose the algorithm accepts. By Theorem~\ref{thm:abelian},
$([n],A)$ is an abelian group, and by Theorem~\ref{thm:bilinear}, multiplication
is biadditive. The final step verifies
\[
F(g_i,g_j,g_\ell)=0
\qquad
\text{for all }i,j,\ell\in[k].
\]
By the triadditivity of $F$, this implies $F(x,y,z)=0$ for all $x,y,z\in G$.
Thus multiplication is associative, and all ring axioms hold.

Finally, the first two steps take $O(n^2)$ time, while the associativity test
uses only $O(k^3)=O(\log^3 n)$ checks. Hence the total running time is
$O(n^2)$.
\end{proof}

\section{Verifying Fields in Linear Time}\label{sec:fields}

Field verification now follows immediately from our ring verifier. After
certifying the ring axioms, it remains only to check that multiplication is
commutative, that there is a multiplicative identity distinct from $0$, and
that every nonzero element has a multiplicative inverse. All of these
conditions can be checked directly in $O(n^2)$ time.

\begin{theorem}[Field verification]\label{thm:field}
Algorithm~\ref{alg:field} accepts if and only if $([n],A,M)$ is a field.
It runs deterministically in $O(n^2)$ time.
\end{theorem}

\begin{proof}
If $([n],A,M)$ is a field, then it is a ring, so
\textsc{VerifyRing} accepts by Theorem~\ref{thm:main}. The remaining checks
succeed because multiplication is commutative, has an identity $1\neq 0$, and
every nonzero element has a multiplicative inverse.

Conversely, suppose the algorithm accepts. By Theorem~\ref{thm:main},
$([n],A,M)$ is a ring. The remaining checks certify that multiplication is
commutative, has an identity $1\neq 0$, and that every $x\neq 0$ has some
$y$ with $xy=1$. By commutativity, also $yx=1$, so every nonzero element is
invertible. Hence $([n],A,M)$ is a field.

Finally, \textsc{VerifyRing} takes $O(n^2)$ time. Finding the multiplicative
identity, checking commutativity, and checking inverses each take at most
$O(n^2)$ time. Therefore the total running time is $O(n^2)$.
\end{proof}

\begin{algorithm}[t]
\caption{\textsc{VerifyField}$(A,M)$}
\label{alg:field}
\begin{algorithmic}[1]
\State run \textsc{VerifyRing}$(A,M)$; \textbf{reject} if it rejects
\State let $0$ be the additive identity certified by the internal call to
\textsc{VerifyAbelian}
\State find $1\in[n]$ such that $M(1,x)=M(x,1)=x$ for every $x\in[n]$;
\textbf{reject} if no such $1$ exists or if $1=0$
\State \textbf{reject} unless $M(x,y)=M(y,x)$ for every $x,y\in[n]$
\For{$x\in[n]\setminus\{0\}$}
    \State \textbf{reject} unless there exists $y\in[n]$ such that $M(x,y)=1$
\EndFor
\State \textbf{accept}
\end{algorithmic}
\end{algorithm}

\section*{Acknowledgements}
Ding was supported in part by a grant from the Israel Science Foundation (ISF Grant No. 1774/20), and by the European Union (ERC, SCALE,101162665). 
Views and opinions expressed are
however those of the author(s) only and do not necessarily reflect those of the
European Union or the European Research Council. Neither the European Union
nor the granting authority can be held responsible for them. 

\section*{AI Disclosure}
We used ChatGPT to assist with the technical presentation, organization, and refinement of the descriptions of our algorithms, with substantial use in Sections 4--6. All AI-assisted content was carefully reviewed and verified by the authors. The authors take full responsibility for the correctness, originality, and integrity of the manuscript, including all reported results and findings.

\bibliographystyle{alpha}
\bibliography{ref}

\end{document}